\documentclass{article}

\usepackage{xspace}
\usepackage{tikz}
\usetikzlibrary{arrows,backgrounds,automata,shapes,patterns}

\usepackage{multirow,tabularx}
\usepackage{verbatim}
\usepackage{authblk} 
\usepackage{enumitem}
\usepackage{booktabs}
\usepackage{latexsym}
\usepackage{amssymb}
\usepackage{amsmath}
\usepackage{amsthm}
\usepackage{stmaryrd}
\usepackage{flushend} 
\usepackage{bbm}
\usepackage{times}

\newtheorem{theorem}{Theorem}
\newtheorem{proposition}[theorem]{Proposition}
\newtheorem{lemma}[theorem]{Lemma}
\newtheorem{corollary}[theorem]{Corollary}

\newtheorem{example}{Example}
\newtheorem{definition}{Definition}

\newcommand{\adom}{\textit{adom}}

\renewcommand{\leq}{\leqslant} 
\renewcommand{\geq}{\geqslant}
\renewcommand{\phi}{\varphi}

\newtheorem{question}[theorem]{Question}

\newcommand{\D}{\mathcal D}
\newcommand{\N}{\mathcal N}
\newcommand{\G}{\mathcal G}
\renewcommand{\L}{\mathcal L}
\newcommand{\F}{\mathcal F}

\newcommand{\myi}{{(i)}\xspace}
\newcommand{\myii}{{(ii)}\xspace}

\newcommand{\argmin}{\operatornamewithlimits{argmin}}

\begin{document}

\title{Database Aggregation}

\author[1]{Francesco Belardinelli}
\author[2]{Umberto Grandi}

\affil[1]{IBISC, Universit\'e d'Evry, France}
\affil[2]{IRIT, University of Toulouse, France}

\maketitle

\begin{abstract}
Knowledge can be represented compactly in a multitude ways, from a set
of propositional formulas, to a Kripke model, to a database. In this
paper we study the aggregation of information coming from multiple
sources, each source submitting a database modelled as a first-order
relational structure. In the presence of an integrity constraint, we
identify classes of aggregators that respect it in the aggregated
database, provided all individual databases satisfy it. We also
characterise languages for first-order queries on which the answer to
queries on the aggregated database coincides with the aggregation of
the answers to the query obtained on each individual database.  This
contribution is meant to be a first step on the application of
techniques from rational choice theory to knowledge representation in
databases.
%
%
\end{abstract}

\section{Introduction}

Aggregating information coming from multiple sources is a
long-standing problem in both knowledge representation and the study
of multi-agent systems (see, e.g., \cite{vanHarmelen2007}).  Depending
on the chosen representation for the incoming pieces of knowledge or
information, a number of competing approaches has seen the light in
these literatures.  Belief
merging \cite{LiberatoreSchaerf1998,KoniecznyPinoPerezJLC2002,KoniecznyLangMarquisAIJ2004}
studies the problem of aggregating propositional formulas coming from
a number of different agents into a set of models, subject to an
integrity constraint.  Judgment and binary
aggregation \cite{EndrissHBCOMSOC2016,DokowHolzmanJET2010,GrandiEndrissAIJ2013}
asks individual agents to report yes/no opinions on a set of logically
connected binary issues, called the agenda, to take a collective
decision.  Social welfare functions, the cornerstone problem in social
choice theory (see, e.g., \cite{Arrow1963}), can also be viewed as
mechanisms merging conflicting information, namely the individual
preferences of voters expressed in the form of linear orders over a
set of alternatives.  Other examples include graph
aggregation \cite{EndrissGrandiAIJ2017}, which has applications in
multi-agent
argumentation \cite{BoothEtAlKR2014,CaminadaPigozzi2011,ChenEndrissTARK2017}
and clustering aggregation \cite{GionisEtAlTKDD2007}, as well as
ontology merging \cite{Porello2014}.

In this work we take a general perspective and we represent individual
knowledge coming from multiple sources as a profile of databases,
modelled as finite relational
structures \cite{AbiteboulHV95,MaierUV84}.  Our aim is to reconcile
two possibly conflicting views of the problem of information fusion.
On the one hand, the study of information merging (typically knowledge
or beliefs) in knowledge representation has focused on the design of
rules that guarantee the consistency of the outcome, with the main
driving principles inspired from the literature on belief
revision.\footnote{Albeit we acknowledge the work
of \cite{DoyleWellmanAIJ1991, MaynardZhangLehmannJAIR2003}, which
aggregate individual beliefs, modelled as plausibility orders, in an
"Arrovian" fashion.}.  On the other hand, social choice theory has
focused on agent-based properties, such as fairness and
representativity of an aggregation procedure, paying attention as well
on possible strategic behaviour by either the agents involved in the
process or an external influencing source.  While there already have
been several attempts at showing how specific merging or aggregation
frameworks could be simulated or subsumed by one another (see,
e.g., \cite{GrandiEndrissIJCAI2011,DietrichList2007a,GregoireKonieczny2006,EveraereEtAl2015}),
we believe that a more general perspective is needed to reconcile the
two views described above.  Perhaps the closest approach to ours is
the work of Baral \emph{et al.} \cite{BaralEtAl1992}.  In their paper,
the authors consider the problem of merging information represented in
the form of a first-order theory, taking a syntactic rather than a
semantic approach, and focuses on finding maximally consistent sets of
the union of the individual theories received.  In doing so, however,
the authors privilege the knowledge representation approach, and have
no control on the set of agents supporting a given maximally
consistent set rather than another.

Our starting point is a set of finite relational structures on the
same signature, coming from a set of agents or sources, and our
research problem is how to obtain a collective databases summarising
the information received. Virtually all of the settings mentioned
above (beliefs, graphs, preferences, judgments...) can be represented
as databases, showing the generality of our framework.  We propose a
number of rules for database aggregation, inspired by existing ones
proposed in the literature on computational social choice, and we
evaluate them axiomatically.  We privilege computationally friendly
aggregators, for which the time to determine the collective outcome is
polynomial in the time spent reading the individual input received.  

When integrity constraints are present, we study how to guarantee that a given aggregators ``lifts'' the integrity constraint from the individual to the collective level, i.e., the
aggregated databases satisfies the same constraints as the individual
ones. We first analyse the problem of lifting first-order formulas in database
aggregation theoretically, comparing the results obtained with the literature on
lifting propositional constraints in binary aggregation.
We provide characterisation results for a number of natural restricted languages, and we investigate
which of the rules we introduced lift classical integrity constraints from database theory: 
functional dependencies, referential integrity constraints, and value constraints.

Since databases are typically queried using formulas in first
order logic, a natural question to ask in a multi-agent setting is
whether the aggregation of the individual answers to a query coincides
with the answer to the same query on the aggregated database.
We provide a partial answer to this important problem, by identifying sufficient conditions on the
first-order query language for both the intersection and the union
operator. 

The paper is organised as follows. In Section~\ref{sec:preliminaries}
we introduce the basic definitions of databases and integrity
constraints. In Sections~\ref{aggregators} and~\ref{sec:axioms} we introduce a number of
database aggregation procedures, and we propose axiomatic properties
for their studies. Sections~\ref{sec:lifting}, ~\ref{sec:characterisations}, and ~\ref{sec:queries} contains our main results on the lifting of integrity constraints and aggregated query
answering.  Section~\ref{sec:conclusions} concludes the paper.

\section{Preliminaries on Databases}\label{sec:preliminaries}

In this section we introduce basic notions on databases that we will
use in the rest of the paper. In particular, we adopt a relational
perspective~\cite{AbiteboulHV95} and define a database as a finite
relational structure over a database schema:
\begin{definition}[Database Schema]
A {\em (relational) database schema} $\D$ is a finite set $\{ P_1 /
q_1,\dots,P_n / q_n \}$ of relation symbols $P$ with arity $q \in
\mathbb{N}$.
\end{definition}

In the following we assume a countable domain $U$ of elements $u, u',
\ldots$, for the interpretation of relation symbols in a database
schema $\D$.
\begin{definition}[Database Instance] \label{dbinstance}
Given domain $U$ and database schema
$\D$, a {\em $\D$-instance} over $U$ is a mapping $D$ associating each
relation symbol $P \in \D$ with a finite $q$-ary relation over $U$,
i.e., $D(P) \underset{{\small fin}}{\subset} U^{q}$.
\end{definition}

By Def.~\ref{dbinstance} a database instance is a finite (relational)
model of a database schema.  The {\em active domain} $\adom(D)$ of an
instance $D$ is the set of all individuals in $U$ occurring in some
tuple $\vec{u}$ of some predicate interpretation $D(P)$,
that is, $\adom(D) = \bigcup_{P \in \D} \{ u \in U \mid u = u_i \text{
  for some } \vec{u} \in D(P) \}$.  Observe that, since $\D$ contains
a finite number of relation symbols and each $D(P)$ is finite, so is
$\adom(D)$. We denote the set of all instances on $\D$ and $U$ as
$\D(U)$. Clearly, the formal framework for databases we adopt is quite
simple, but still it is powerful enough to cover practical cases of
interest \cite{MaierUV84}. Here we do not discuss in details the pros and cons of the
relational approach to database theory and refer to the literature for
further details \cite{AbiteboulHV95}.

To specify the properties of databases, we make use of first-order
logic with equality and no function symbols.  Let $V$ be a countable
set of {\em individual variables}. 
\begin{definition}[FO-formulas over $\D$]\label{def:fo}
Given a database schema $\D$, the formulas $\varphi$ of the
first-order language $\L_{\D}$ are defined by the following BNF:
\begin{eqnarray*}
\varphi & ::= & x = x'\mid P(x_1, \ldots ,x_{q}) \mid \lnot \varphi
\mid \varphi \to \varphi \mid \forall x \varphi
\end{eqnarray*}
where $P \in \D$, $x_1, \ldots ,x_{q}$ is a $q$-tuple of terms and $x,
x' $ are terms.
\end{definition}
We assume ``$=$'' to be a special binary predicate with fixed obvious
interpretation. By Def.~\ref{def:fo}, $\L_\D$ is a first-order
language with equality over the relational vocabulary $\D$ and with no
function symbols.
%
In the following we use the standard abbreviations $\exists$, 
$\wedge$,
$\vee$, and $\neq$.
Also, free and bound variables are defined as standard.  For a formula
$\varphi \in \L_{\D}$,
we write $\varphi(x_1,\ldots,x_\ell)$, or simply $\varphi(\vec x)$, to
list explicitly in arbitrary order all free variables
$x_1,\ldots,x_\ell$ of $\varphi$.
A {\em sentence} is a formula with no free variables.  Notice that the
only terms in our language $\L_{\D}$ are individual variables. We can
add constant for individuals with some minor technical changes to the
definitions and results in the paper. However, these do not impact on
the theoretical contribution and we prefer to keep notation lighter.

To interpret FO-formulas on database instances, we introduce 
{\em assignments} as functions $\sigma: V \mapsto U$.
Given an assignment $\sigma$, we denote by $\sigma^x_u$ the 
assignment such that
\myi $\sigma^x_u(x) = u$; and \myii 
$\sigma^x_u(x') = \sigma(x')$, for every $x' \in V$
different from $x$.
We can now define the semantics of $\L_\D$.
\begin{definition}[Satisfaction of FO-formulas]\label{def:fo-sem}
Given a $\D$-instance $D$, an assignment $\sigma$, and an FO-formula
$\varphi\in\L_{\D}$, we inductively define whether $D$ \emph{satisfies
  $\varphi$ under $\sigma$}, or $ (D, \sigma) \models \varphi$, as
follows:
\begin{tabbing}
 $ (D, \sigma)\models P(x_1,\ldots,x_{q})$ \ \ \ \=
 iff \ \ \= $\langle \sigma(x_1),\ldots,\sigma(x_{q}) \rangle \in D(P)$\\ 
 $ (D, \sigma)\models x = x'$ \> iff \> $\sigma(x)=\sigma(x')$\\
 $ (D, \sigma)\models \lnot\varphi$ \> iff \> $(D, \sigma) \not \models\varphi$\\
 $ (D, \sigma)\models \varphi \to \psi$ \> iff \> $(D,\sigma) \not \models \varphi $ or $ (D, \sigma)\models\psi$\\
 $ (D, \sigma)\models \forall x\varphi$ \> iff \> for every
 $u\in \adom(D)$, $ (D, \sigma^x_u) \models\varphi$
\end{tabbing}

A formula $\varphi$ is {\em true} in $D$, written $D\models\varphi$,
iff $ (D, \sigma) \models \varphi$, for all assignments $\sigma$.
\end{definition}
Observe that we adopt an {\em active-domain} semantics, that is,
quantified variables range only over the active domain of $D$.  This
is standard in database theory \cite{AbiteboulHV95}, where $\adom(D)$
is assumed to be the ``universe of discourse''.\\

\textbf{Constraints.}
It is well-known that several properties and constraints on databases
can be expressed as FO-sentences. Here we consider some of these for
illustrative purposes.
\begin{definition}[Functional Dependency]\label{def:functionaldependency}
A {\em functional dependency} is an expression of type $n_1, \ldots,
n_k \mapsto n_{k+1}, \ldots, n_{q}$.  A database instance $D$
satisfies a functional dependency $n_1, \ldots, n_k \mapsto n_{k+1},
\ldots, n_{q}$ for predicate symbol $P$ with arity $q$ iff for every
$q$-ple $\vec{u}$, $\vec{u}'$ in $D(P)$, whenever $u_i = u'_i$ for all $i \leq k$,
then we also have $u_i = u'_i$ for all $k+1 \leq i \leq q$.
If $k = 1$, we say that it is a {\em key dependency}.
\end{definition}

Clearly, any database instance $D$ satisfies a functional dependency
$n_1, \ldots, n_k \mapsto n_{k+1}, \ldots, n_{q}$ iff it satifies the
following:
\begin{eqnarray*}
 \forall \vec{x} , \vec{y} \left(P(\vec{x}) \land P(\vec{y}) \land
 \bigwedge_{i \leq k }(x_i = y_i) \to \bigwedge_{k+1 \leq i \leq
   q}(x_i = y_i)\right)
\end{eqnarray*}

\begin{definition}[Value Constraint]
A {\em value constraint} is an expression of type $n_k \in D(P_v)$, where $D(P_v)$ contains a list of admissible values.  A
database instance $D$ satisfies a value constraint $n_k \in P_v$ for predicate symbol $P$ with arity $q \geq k$ iff for every $q$-ple $\vec{u}$ in $D(P)$, $u_k \in D(P_v)$.
\end{definition}

Also for value constraints, it is easy to check that an instance $D$
satisfies constraint $n_k \in P_v$ for symbol $P$ iff it satisfies the following:
\begin{eqnarray*}
\forall x_1, \dots, x_q (P(x_1,\dots,x_q) \to P_v(x_k))
\end{eqnarray*}

\begin{definition}[Referential Integrity Constraint]
A referential integrity constraints enforces the foreign key of a predicate $P_1$ to be the primary key of predicate $P_2$.
A database instance satisfy a referential integrity constraint on the last $k$ attributes, and we denote it $(P_1\to P_2, k)$, if for all $q_{1}$-uple $\vec{u}\in D(P_1)$, there exists a $q_{2}$-uple $\vec{u}'\in D(P_2)$ such that for all $1\leq i \leq k$ we have that $u_{q_1-k+j}=u_j'$.
\end{definition}

A referential integrity constraint can also be translated in a first-order formula as follows:
\begin{eqnarray*}
\forall \vec{x} [ P_1(\vec{x}) \rightarrow \exists \vec{y} (P_2(\vec{y}) \wedge \bigwedge_{i=1}^{k} (x_{q_1-k+j}=y_j)) ]
\end{eqnarray*}

%


\section{Aggregators} \label{aggregators}

The main research question we investigate in this paper regards how to
define an aggregated database instance from the instances of
$\N=\{1,\dots,n\}$ agents.  This question is typical in social choice
theory, where judgements, preferences, etc., are aggregated according
to some notion of rationality that will be introduced in
Section~\ref{sec:lifting}.

For the rest of the paper we fix a database schema $\D$ over a common
domain $U$, and consider a {\em profile } $\vec{D} = (D_1, \dots,
D_n)$ of $n$ instances over $\D$ and $U$.  Then, we can define an
aggregation procedure on such instances.
\begin{definition}[Aggregation Procedure]
Given database schema $\D$ and domain $U$, an {\em aggregation
  procedure} $F: \D(U)^n \to \D(U)$ is a function assigning to each
  tuple $\vec{D}$ of instances for $n$ agents an aggregated instance
  $F(\vec{D}) \in \D(U)$.  Let $\F$ be the class of all aggregation
  procedures.
\end{definition}

We use $N_{\vec{u}}^{\vec{D}(P)} :: = \{ i \in \N \mid \vec{u} \in
D_i(P) \}$ to denote the set of agents accepting tuple $\vec{u}$ for
symbol $P$, under profile $\vec{D}$. Notice that considering a unique
domain $U$ is not really a limitation of the proposed approach:
instances $D_1, \ldots, D_n$, each on a possibly different domain
$U_i$, for $i \leq n$, can all be seen as instances on
$\bigcup_{i \in \N} U_i$.

Hereafter we illustrate and discuss some examples of aggregation
procedures:
%

\textbf{Union} (or nomination): for every $P \in \D$, $F(\vec{D})(P)=\bigcup_{i
  \leq n} D_i(P) $. 
Intuitively, every agent is seen as having partial but correct
  information about the state of the world. Union can be considered a
  good aggregator if databases represent the agents' knowledge bases
  (certain information).
\smallskip

\textbf{Intersection} (or unanimity): for every $P \in \D$,
  $F(\vec{D})(P)=\bigcap_{i \leq n} D_i(P) $. 
Here every agent is supposed to have a partial and possibly incorrect
  vision of the state of the world.
\smallskip

\textbf{Quota rules}: a {\em quota} rule is an aggregation rule $F$
  defined via functions $q_P : U^q \to \{0,1, \ldots , n+1 \}$,
  associating each symbol $P$ and $q$-uple with a quota, by
  stipulating that $\vec{u} \in F(\vec{D})(P)$ iff $|\{i \mid \vec{u}\in D_i(P)\}| \geq
  q_P(\vec{u})$.  $F$ is called {\em uniform} whenever $q$ is a
  constant function for all tuples and symbols.  Intuitively, if a
  tuple $\vec{u}$ appears in at least $q(\vec{u})$ of the initial
  databases, then it is accepted.
The (strict) majority rule is a quota rule for $q = \lceil (n+1)/2
\rceil$; while  union and intersection are quota rule for $q = 1$ and $q = n$ respectively. We call the uniform quota rules for $q = 0$ and $q =
n + 1$ {\em trivial rules}.
\smallskip

\textbf{Distance-based function}: 
%
%
The symmetric distance can be used to measure dissimilarity between databases, obtaining the following definition:
\begin{eqnarray*}
F(\vec{D})(P) & = & \argmin_{A \underset{fin}{\subset} U^{q_P}} \sum_{i \in \N} (|D_i(P)\setminus A| + |A \setminus D_i(P)|)
\end{eqnarray*}

Intuitively, the symmetric distance minimizes the ``distance'' between
the aggregated database $F(\vec{D})$ and each $D_i$, defined as the
number of tuples in $D_i$ but not in $F(\vec{D})$, plus the number of
tuples in $F(\vec{D})$ but not in $D_i$, calculated across all
$i \in \N$.


\smallskip

\textbf{Dictatorship of agent $i^* \in \N$}: 
we have that $F(\vec{D}) = D_{i^*}$, i.e., the dictator $i^*$
  completely determines the aggregated database.
\smallskip

{\bf Oligarchy of coalition $C^* \subseteq \N$}: for every $P \in \D$,
  $F(\vec{D})(P) = \bigcap_{i \leq C^*} D_{i}(P) $.  Oligarchy reduces
  to dictatorship for singletons, and to intersection for $C^* = \N$.
\smallskip



 Quota rules are inspired by their homonyms in judgment
aggregation \cite{DietrichListJTP2007}, introduced as a generalisation
of the classic majority rule. The union and the intersection rules are
well-known in the area of modal epistemic logic, corresponding,
respectively, to distributed knowledge and ``everybody knows
that'' \cite{Hintikka1962}. Distance-based procedures have been widely
studied and axiomatised in the area of logic-based belief
merging \cite{KoniecznyPinoPerezJLC2002}, while dictatorships and
oligarchies are classical notions from social choice theory.
Obviously, different aggregation procedures can be thought of. We
chose to focus on those above in the following, as they are
well-studied in the literature and have nice computational properties
such as being computable in polynomial time. 

\section{The Axiomatic Method}\label{sec:axioms}

Aggregation procedures are best characterised by means of axioms. In
particular, we consider the following properties, where relation
symbols $P, P' \in \D$, profiles $\vec{D}, \vec{D}' \in \D(U)^n$,
tuples $\vec{u}$, $\vec{u}' \in U^+$ are all universally quantified.
%

\smallskip

\textbf{Independence ($I$)}: 
if $N_{\vec{u}}^{\vec{D}(P)} = N_{\vec{u}}^{\vec{D}'(P)}$ then
$\vec{u} \in F(\vec{D})(P)$ iff $\vec{u} \in F(\vec{D}')(P)$.

\smallskip

Intuitively, if the same agents accepts (resp.~reject) a tuple in two
different profiles, then the tuple is accepted (resp.~rejected) in
both aggregated instances.  The axiom of independence is a widespread
requirement from social choice theory, and is arguably the main cause
of most impossibility theorems, such as Arrow's seminal
result \cite{Arrow1963}.  From a computational perspective,
independent rules are typically easier to compute than non-independent
ones.  Clearly, quota rules satisfy independence; while neither
dictatorship nor oligarchies do.

\smallskip

\textbf{Unanimity ($U$)}: 
$F(\vec{D})(P) \supseteq \bigcap_{i \in \N} D_i(P)$.

\smallskip

That is, a tuple accepted by all agents, also appears in the
aggregated database (for the relevant relation symbol). In
particular, all rules in Section~\ref{aggregators} satisfy unanimity.
%

\smallskip

\textbf{Groundedness ($G$)}: 
$F(\vec{D})(P) \subseteq \bigcup_{i \in \N}
  D_i(P)$.

\smallskip

By groundedness any tuple appearing in the aggregated database must be
accepted by some agent. All rules from Section~\ref{aggregators}, with
the exception of the distance-based rule, satisfy this property.

\smallskip

\textbf{Anonymity ($A$)}: for every 
permutation $\pi : \N \to \N$, we have $F(D_1, \ldots , D_n ) = F
(D_{\pi(1)} , \ldots , D_{\pi(n)})$.

\smallskip

Here the identity of agents is irrelevant for the aggregation
procedure. Clearly, this is the case for all aggregators in
Section~\ref{aggregators} but dictatorship and oligarchy.

\smallskip

\textbf{Positive Neutrality ($N^+$)}: 
if $N_{\vec{u}}^{\vec{D}(P)} = N_{\vec{u}'}^{\vec{D}(P)}$ then
$\vec{u} \in F(\vec{D})(P)$ iff $\vec{u}' \in F(\vec{D})(P)$.

\textbf{Negative Neutrality ($N^-$)}: 
if $N_{\vec{u}}^{\vec{D}(P)} = \N \setminus N_{\vec{u}'}^{\vec{D}(P)}$ then
$\vec{u} \in F(\vec{D})(P)$ iff $\vec{u}' \not\in F(\vec{D})(P)$.

Observe that both versions of neutrality differs from independence as here we consider
two different tuples in the same profile, while independence deals
with the same tuple in two different profiles.
We can easily see that all aggregators introduced in Section~\ref{aggregators} satisfy positive neutrality and, with the exception of most quota rules (see below), negative neutrality as well.
%

\smallskip

\textbf{Systematicity ($S$)}: 
if $N_{\vec{u}}^{\vec{D}(P)} = N_{\vec{u}'}^{\vec{D}(P')}$
then
  $\vec{u} \in F(\vec{D})(P)$ iff $\vec{u}' \in F(\vec{D})(P')$.

\smallskip

Observe that systematicity is equivalent to the conjunction of neutrality and independence.



\smallskip

\textbf{Permutation-Neutrality ($N^{P}$)}: Given a permutation $\rho
  : U \to U$ over domain $U$, and its straightforward lifting 
  to a profile $\vec{D}$, then
  $F(\rho(\vec{D}))=\rho(F(\vec{D}))$.

\smallskip

Again, all aggregators but dictatorship and oligarchies satisfy
permutation-neutrality.


 


\textbf{Monotonicity ($M$)}: 
if $\vec u \in F(\vec{D})(P)$ and for every $i \in \N$, either
$D_i(P) = D'_i(P)$ or $D_i(P) \cup \{\vec{u}\} \subseteq D'_i(P)$, then $\vec
u \in F(D')(P)$.





\smallskip

Intuitively, a monotonic aggregators keeps on accepting a given tuple
if the support for that tuple increases.

\smallskip

Combinations of the axioms above can be used to characterise some of
the rules that we defined in Section~\ref{aggregators}. Some of these
results, such as the following, lift to databases known results in
judgement (propositional) aggregation.
\begin{lemma}
An aggregation procedure satisfies $A$, $I$, and $M$ 
iff it is a quota rule.
\end{lemma}
\begin{proof}
The implication from right to left follows from the fact that quota
rules satisfy independence $I$, anonymity $A$, and monotonicity $M$,
as we remarked above.

For the implication from left to right, observe that, to accept a
given tuple $\vec{u}$ in $F(\vec{D})(P)$, an independent aggregation
procedure will only look at the set of agents $i \in \N$ such that
$\vec{u} \in D_i(P)$. If the procedure is also anonymous, then
acceptance is based only on the number of individuals admitting the
tuple. Finally, by monotonicity, there will be some minimal number of
agents required to trigger collective acceptance. That number is the
quota associated with the tuple and the symbol in hand.
\end{proof}

If we add neutrality (both positive and negative), then we obtain the
class of uniform quota rules. If we furthermore impose unanimity and
groundedness, then this excludes the trivial quota rules.
\begin{lemma}
If the number of individuals is odd and $|\D| \geq 2$, an aggregation
procedure $F$ satisfies $A$, $N^-$, $N^+$, $I$ and
$M$ on the full domain $\D(U)^n$ if and only if it is the majority
rule.
\end{lemma}
\begin{proof}
By neutrality the quota must be the same for all tuples and all
relation symbols. By negative-neutrality the two sets
$N_{\vec{u}}^{\vec{D}(P)}$ and $\N \setminus N_{\vec{u}}^{\vec{D}(P)}$
must be treated symmetrically.  Hence, the only possibility is to have
a uniform quota of $(n +1)/2$.
\end{proof}

The corresponding versions of these results have been shown in
judgment and graph
aggregation \cite{DietrichListJTP2007,EndrissGrandiAIJ2017}.  Notice
however that there are some notable differences w.r.t.~the
literature. For instance, the axiom of neutrality is here split into a
positive and a negative part.

We conclude this section by showing the following equivalence between majority and distance-based rules.
\begin{lemma}
In the absence of integrity constaints, and for an odd number of
agents, the distance-based rule coincides with the majority rule.
\end{lemma}

\begin{proof}
By the definition of the distance based rule, we have that 
\begin{eqnarray*}
F(\vec{D})(P) & = & \argmin_{A \underset{fin}{\subset} U^{q_P}} \sum_{i \in \N} (|D_i(P)\setminus A| + |A \setminus D_i(P)|)
\end{eqnarray*}

With a slight abuse of notation, if $A \subseteq U^{m}$ let
$A(\vec{u})$ be its characteristic function.  Since the minimisation
is not constrained, and all structures are finite, this is equivalent
to:
\begin{eqnarray*}
F(\vec{D})(P) & = & \argmin_{A \underset{fin}{\subset} U^{q_P}} \sum_{i \in \N}  \sum_{\vec{u}\in U^{q_P}} |D_i(P)(\vec{u}) - A(\vec{u})| \\
& = & \argmin_{A \underset{fin}{\subset} U^{q_P}}  \sum_{\vec{u}\in U^{q_P}} \sum_{i \in \N}  |D_i(P)(\vec{u}) - A(\vec{u})|
\end{eqnarray*}
Therefore, for each $\vec{u}$, if for a majority of the individuals in $\N$ we have that $\vec{u}\in D_i(P)$, then $\vec{u}\in A$ minimises the overall distance, and symmetrically for the case in which a majority of individuals are such that $\vec{u}\not \in D_i(P)$. \end{proof}

\section{Lifting Constraints}\label{sec:lifting}


In this section we analyse further the properties of the aggregation
procedures introduced in Section~\ref{aggregators}. Specifically, we
present a notion of {\em collective rationality} that aims to capture
the appropriateness of a given aggregator $F$ w.r.t.~some constraint
$\phi$ on the input instances $D_1, \ldots, D_n$.
Hereafter let $\phi$ be a sentence in the first-order language
$\L_{\D}$ associated to $\D$, interpreted as a common constraint that
is satisfied by all $D_1, \ldots, D_n$. Here we are interested in the
following notion:
\begin{definition}[Collective Rationality]
A constraint $\phi$ is lifted by an aggregation procedure $F$ if
whenever $D_i \models \phi$ for all $i \in \N$, then also
$F(\vec{D}) \models \phi$. 

 An aggregation procedure $F :
\D(U)^n \to \D(U)$ is {\em collectively rational} (CR) with respect to
$\phi$ iff $F$ lifts $\phi$.
\end{definition}

Intuitively, an aggregator is CR w.r.t.~constraint $\phi$ iff it
lifts, or preserves, $\phi$.
\begin{example}
 We now provide an illustrative example of first-order collective
(ir)rationality with the majority rule.  Consider agents 1 and 2 with
database schema $\D = \{ P/1, Q/2 \}$.  Two database instances are
given as $D_1 = \{ D(a), Q(a,b) \}$ and $D_2 = \{ D(a),
Q(a,c) \}$. Clearly, both instances satisfy the integrity constraint
$\phi = \forall x (P(x) \to \exists y Q(x,y))$.  However, their
aggregate $D = F(D_1,D_2) = \{ D(a) \}$, obtained by the majority
rule, does not satisfy $\phi$.
This example, which can be considered a {\em paradox} in the sense
of \cite{GrandiEndrissAIJ2013}, shows that not every constraint in the
language $\L_{\D}$ is collective rational w.r.t.~mojority, thus
obtaining a first, simple negative result.
\end{example}

One natural question to ask about lifting of constraints is the following.
\begin{question}
Given an axiom AX, what is the class of constraints that are lifted by
all aggregators $F$ satisfying AX?
\end{question}

To make this question more precise, consider the following definition.
\begin{definition}
Given a language $\L \subseteq \L_{\D}$, define $CR[\L]$ as the class
of aggregation procedures that lift all $\phi \in \L$:
\begin{eqnarray*}
CR[\L] & ::= & \{F : D(U)^n \to D \mid \text{$F$ is CR for all
$\phi \in \L$}\}
\end{eqnarray*}

Moreover, an aggregator $F$ {\em satisfies a set $AX$ of axioms
w.r.t.~language $\L$}, if $F$ satisfies the axioms in AX on set $\{
D \in \D(U)
\mid D \models \phi \}$ for all constraints $\phi \in \L$. 
The class of all such aggregators is given as:
\begin{eqnarray*}
\F_{\L}[AX] & := & \{F : \D(U)^{n} \to \D(U) \mid \text{
$F$ satisfies $AX$ on} \\ 
& & \text{$\{ D \in \D(U) \mid D \models \phi \}$ for all $\phi \in \L$}\}
\end{eqnarray*}
\end{definition}

The following Lemmas extend results in \cite{GrandiEndrissAIJ2013} to
the case of database aggregation.  Hereafter, for a language $\L$ and
operator $\bullet$, $\L^{\bullet}$ is the language obtained by closing
formulas in $\L$ under $\bullet$.  The proofs are immediate, so we
omit them. We only remark that point (3) follows from the fact that
the constraints $\phi \in \L$ are assumed to be sentences.
\begin{lemma} \label{aux1}
 For every language $\L \subseteq \L_{\D}$:
\begin{enumerate}
\item $CR[\L^{\land}] = CR[\L^{\equiv}] = CR[\L]$


\item $CR[\L \cup \{ \top \}] = CR[\L \cup \{ \bot \}] = CR[\L]$

\noindent
Moreover, 

\item $CR[\L^{\forall}] = CR[\L^{\exists}] = CR[\L]$
\end{enumerate}
\end{lemma}

%

By Lemma~\ref{aux1} an aggregator $F$ is CR w.r.t.~a language $\L$ iff
it is CR w.r.t.~the closure of $\L$ under either conjuction, or
coimplication, or universal or existential quantification. Also,
adding either $\top$ or $\bot$ does not change collective rationality.

Furthermore, the following result extends Lemma 7
in \cite{GrandiEndrissAIJ2013}. Also in this case, proofs are
immediate and therefore omitted.
\begin{lemma} \label{aux2}
For all languages $\L_1, \L_2 \subseteq \L_{\D}$,
\begin{enumerate}
\item If $\L_1 \subseteq \L_2$ then $CR[\L_2] \subseteq CR[\L_1]$
\item $CR[\L_1 \cup \L_2] = CR[\L_1] \cap CR[\L_2]$
\end{enumerate}
\end{lemma}
%
%

By Lemma~\ref{aux2}, collective rationality is anti-monotone
w.r.t.~language inclusion, and an aggregator $F$ is CR w.r.t.~the
union of languages iff it is CR w.r.t.~each language separately.

The next results, which extend Lemma 8 in \cite{GrandiEndrissAIJ2013}, relate collective rationality with axioms.
\begin{lemma} \label{lemma8}
For all languages $\L_1, \L_2 \subseteq \L_{\D}$,
\begin{enumerate}
\item If $\L_1 \subseteq \L_2$ then $\F_{\L_2}[AX] \subseteq \F_{\L_1}[AX]$\\
In particular, if $\top \in \L$ then $\F_{\L}[AX] \subseteq \F_{\{ \top \}}[AX]$
\item $\F_{\L}[AX1, AX2] = \F_{\L}[AX1] \cap \F_{\L}[AX2]$
\end{enumerate}
\end{lemma}
\begin{proof}
As regards (1), if $F$ satisfies $AX$ on $\{ D \in \D(U) \mid D
\models \phi \}$, for all $\phi  \in \L_2$, and $\L_1 \subseteq \L_2$, then
in particular it satisfies $AX$ on $\{ D \in \D(U) \mid D \models
\phi \}$, for all $\phi  \in \L_1$.
Then, (2) follows immediately from (1), as $\{ \top \} \subseteq \L$.
As for (3), $F$ satisfies $AX1$ and $AX2$ on $\{ D \in \D(U) \mid D
\models \phi \}$, for all $\phi  \in \L$, iff both $F$ satisfies $AX1$ and
$F$ satisfies $AX2$.
\end{proof}

However, not all results available at the propositional level extend
to the first order. In particular, the following result means that
Lemma 6 in \cite{GrandiEndrissAIJ2013} does not lift to the first
order.
\begin{lemma} \label{lemma6}
There exists languages 
$\L_1$ and $\L_2$, both containing $\top$ and $\bot$, such that
$\L_1 \neq \L_2$ but $CR[\L_1 ] = CR[\L_2]$.
\end{lemma}
\begin{proof}
  Consider languages $\L_1 = \{ \bot, \top \}$ and $\L_2
  = \L_1 \cup \{ \forall x P(x) \}$ on $\D = \{ P/1 \}$. By
  Lemma~\ref{aux2}.(1), $CR[\L_2] \subseteq CR[\L_1]$. Now, suppose
  that $F \in CR[\L_1]$ and consider a profile $\vec{D}$ such that
  $D_i \models \forall x P(x)$ for all $i \in \N$. By definition,
  $F(\vec{D}) \in \D(U)$. We consider two alternatives: either
  $F(\vec{D})$ is empty and then $F(\vec{D}) \models \forall x P(x)$
  trivially; or $F(\vec{D})$ is not empty, then
  $F(\vec{D})(P) \subseteq U$ and $F(\vec{D}) \models \forall x P(x)$
  as well. As a result, $CR[\L_1] \subseteq CR[\L_2]$.
\end{proof}

By Lemma~\ref{lemma6}, the operator $CR[-]$ from languages to sets of
aggregators is not injective in general.

Symmetrically, we introduce an operator $LF[-]$ from sets of
aggregators to languages.
\begin{definition}[Lifted Language] 
 Given a set $\G$ of aggregation procedures, let $LF[\G]$ be the
 language of the constraints that are lifted by all $F \in \G$:
$$LF[\G] ::= \{ \varphi \in \L_{\D} \mid \text{$F$ is $CR$
   w.r.t.~$\varphi$, for all $F \in \G$} \}$$
\end{definition}

Clearly, $LF[\G]$ is the intersection of all $LF[\{ F \}]$, for $F \in
G$.

Lemma~\ref{lemma6} has an impact on the following result, which
correspondent to Proposition 9 in \cite{GrandiEndrissAIJ2013}.  In
particular, while in \cite{GrandiEndrissAIJ2013} we have equality for
item (1), here we only have inclusion.
\begin{proposition}
 Let 
$\L$ a language 
containing $\top$ and $\bot$, and $\G$ a class of aggregators. Then, 
\begin{enumerate}
\item $\L \subseteq LF[CR[\L]]$, and this inclusion is strict for some
  languages.
\item $\G \subseteq CR[LF[G]]$, and this inclusion is strict for some classes.
\end{enumerate}
\end{proposition}
\begin{proof}
As regards (1), inclusion $\L \subseteq LF[CR[\L]]$ is an immediate
consequence of the definitions of $CR$ and $LF$.  On the other hand,
consider languages $\L_1 = \{ \bot, \top \}$ and $\L_2 = \L_1 \cup \{
\forall x P(x) \}$ in the proof of Lemma~\ref{lemma6}. We have
$CR[\L_1] = CR[\L_2]$, and therefore $LF[CR[\L_1]] = LF[CR[\L_2]]$,
but $\L_1 \subset \L_2$, and therefore $\L_1 \subset LF[CR[\L_1]]$.

As for (2), inclusion $\G \subseteq CR[LF[\G]]$ is also an immediate
consequence of the definitions of $CR$ and $LF$. Further,
in \cite{GrandiEndrissAIJ2013} Proposition 9, it is given a class
(basically, $\G$ does not contain generalised dictatorships) for which
this inclusion is strict.
\end{proof}

To conclude, the relationship between operators $CR[-]$ and $LF[-]$
can be represented as in Fig.~\ref{fig1}.
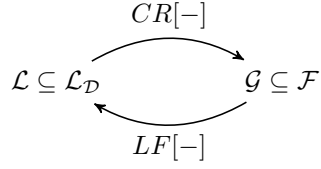
\begin{figure} 
\begin{center}
\begin{tikzpicture}[auto,node distance=3cm,->,>=stealth',shorten
>=1pt,semithick]

\tikzstyle{every state}=[fill=white,text=black,minimum size=0.3cm]
\tikzstyle{every initial by arrow}=[initial text=]

\node[]  (t0) {$\L \subseteq \L_{\D}$};
\node[] (t1) [right of = t0] {$\G \subseteq \F$};
\path 	(t0) 	edge[bend left] node {$CR[-]$}  (t1);
\path 	(t1) 	edge[bend left] node {$LF[-]$}  (t0);
\end{tikzpicture}
\caption{the operators $CR[-]$ and $LF[-]$. \label{fig1}}
\end{center}
\end{figure}
The two operators are inverse one to the other, but they do not
commute.

\section{Characterisation Results}\label{sec:characterisations}

In this section we show some correspondences between axiomatic
properties and restrictions to the first order language in which
integrity constraints can be expressed, in line with previous work by
Grandi and Endriss \cite{GrandiEndrissAIJ2013}. We then focus on
the database-specific constraints introduced in
Section~\ref{sec:preliminaries}, showing sufficient and necessary
conditions for collective rationality of quota rules.

To state the next result we consider a set $Con \subseteq U$ of
constants,
interpreted as themselves in each $D_i$, that is, $\sigma(c) = c$ for
every $c \in Con$. Then, let $lit^+ \subseteq \L_{\D}$ be some language containing
only positive literals of form $P(c_1, \ldots, c_q)$, 
for $P\in \D$ and constants $c_1, \ldots, c_q$.

\begin{theorem} \label{theor1}
$\F_{lit^+}[U] \subseteq CR[lit^+]$, and $\F_{lit^+}[U] \supseteq CR[lit^+]$
only if $Con$ contains all individuals in the domain of $F$.
\end{theorem}
\begin{proof}
As to inclusion $\subseteq$, we see that if all instances
$D_1, \ldots, D_n$ satisfy formulas $P(c_1, \ldots, c_q)$ in $lit^+$,
then $\vec{c} \in D_i(P)$ for every $i \in \N$. By unanimity we have
that $\bigcap_{i \in \N} D_i(P) \subseteq F(\vec{D})(P)$, and
therefore $\vec{c} \in F(\vec{D})(P)$. Hence, $F$ is collectively
rational on $lit^+$.

As to $\supseteq$, suppose that $F \in CR[lit^+]$ and choose a profile
$D_1, \ldots, D_n$ with $\vec{u} \in \bigcap_{i \in \N} D_i(P)$, that
is, for every $i \in \N$, $D_i \models P(u_1, \ldots, u_q)$. Since we
assumed that $Con$ contains all individuals in the domain of $F$,
individuals $u_1, \ldots, u_q$ belong to $Con$ and formulas
$P(u_1, \ldots, u_q)$ are in $lit^+$. Further, $F$ is CR on
$D_1, \ldots, D_n$ and therefore $F(\vec{D}) \models P(u_1, \ldots, u_q)$, that
is, $\vec{u} \in F(\vec{D})(P)$, which mean that $F$ is unanimous.
\end{proof}

By Theorem~\ref{theor1} an aggregator $F$ is collectively rational on
a language $lit^+$ with positive literals only iff it is unanimous on
the class of instances satisfying the very same positive literals.

A symmetric result holds for the axiom of groundedness and any
 language $lit^- \subseteq \L_{\D}$ containing only {\em negative}
 literals of form $\neg P(c_1, \ldots, c_q)$.  The proof is similar,
 so we omit it.
\begin{theorem} \label{theor2}
$\F_{lit^-}[G] \subseteq CR[lit^-]$, and $\F_{lit^-}[G] \supseteq CR[lit^-]$
only if $Con$ contains all individuals in the domain of $F$.
\end{theorem}

From Theorem~\ref{theor1} and \ref{theor2}, we immediately obtain
the following corollary by the lemmas in section~\ref{sec:lifting},
where $lit = lit^+ \cup lit^-$.
\begin{corollary}
$\F_{lit}[U,G] \subseteq CR[lit]$, and $\F_{lit}[U,G] \supseteq CR[lit]$
only if $Con$ contains all individuals in the domain of $F$.
\end{corollary}
\begin{proof}
As to inclusion $\subseteq$, by Lemma~\ref{lemma8}.(2), $\F_{lit}[U,G]
= \F_{lit}[U] \cap \F_{lit}[G]$, and by Lemma~\ref{lemma8}.(1)
$\F_{lit}[U] \cap \F_{lit}[G] \subseteq \F_{lit^+}[U] \cap \F_{lit^-}[G]$. Then,
by Theorem~\ref{theor1} and \ref{theor2},
$\F_{lit^+}[U] \cap \F_{lit^-}[G] \subseteq CR[lit^+] \cap
CR[lit^-]$. Finally, by Lemma~\ref{aux2}.(1) $CR[lit^+] \cap
CR[lit^-] \subseteq CR[lit]$. The other inclusion is proved similarly.
\end{proof}

Notice that, differently from the propositional case \cite[Theorem
10]{GrandiEndrissAIJ2013}, here we need both axioms of unanimity and
groundedness to preserve both positive and negative literals, while
for propositional literals unanimity suffices. Hence, also simple
results do not transfer immediately from the propositional to the
first-order setting.


Next, define $\L_{\leftrightarrow}$ as the language of equivalences
$\forall \vec{x} \vec{x}' (P(\vec{x}) \leftrightarrow P'(\vec{x}'))$
for relation symbols $P, P'
\in \D$. We show the following:

\begin{theorem}\label{theor3}
$ CR[\L_{\leftrightarrow}] = \F_{\leftrightarrow}[N^{+}]$
\end{theorem}
\begin{proof}
 As for inclusion $\supseteq$, pick an equivalence
 $\forall \vec{x}, \vec{x}'(P(\vec{x}) \leftrightarrow
 P'(\vec{x}'))$. This defines a database in which relation symbols $P$
 and $P'$ share the same pattern of acceptance/rejection, and since
 aggregator $F$ is neutral over issues, we get
 $F(\vec{D}) \models \forall \vec{x}, \vec{x}'(P(\vec{x}) \leftrightarrow
 P'(\vec{x}'))$. Therefore, the constraint given by the initial
 equivalence is lifted. \\
 As for inclusion $\subseteq$, suppose that a profile $\vec{D}$ is
 such that $N_{\vec{u}}^{\vec{D}(P)} =
 N_{\vec{u}'}^{\vec{D}(P')}$. This implies that for every $i \in \N$,
 $\vec{D}_i \models \forall \vec{x}, \vec{x}'(P(\vec{x}) \leftrightarrow
 P'(\vec{x}'))$, and since $F$ is in $CR[\L_{\leftrightarrow}]$,
 $\vec{u} \in F(D)(P)$ iff $\vec{u}' \in F(D)(P')$. This holds for
 every such profile $\vec{D}$, proving that $F$ is neutral.
\end{proof}

By Theorem~\ref{theor3} an aggregator $F$ is collectively rational on
language $\L_{\leftrightarrow}$ iff it is positively neutral on the
class of instances satisfying all formulas in $\L_{\leftrightarrow}$.


Let us now define the following class:

\begin{definition}[Generalised dictatorship]
 An aggregation procedure $F : \D(U)^n \to \D(U)$ is a {\em
   generalised dictatorship} if there exists a map $g: \D(U)^n \to \N$
   such that for every $\vec{D} \in \D(U)^n$, $F(\vec{D}) =
   D_{g(\vec{D})}$.  Let $GDIC$ be the class of all generalised
   dictatorships.
\end{definition}

Generalised dictatorships include classical dictatorships, but also
more interesting procedures known as \emph{most representative voters
rules}, which selects the individual input that best summarises a
given profile.  Clearly, since each single instance satisfies the
given set of constraints, a generalised dictatorship is
collectively rational with respect to the full first-order language.
 \begin{theorem} $GDIC \subset CR[\L_{\D}]$ \end{theorem}

Observe that while for binary aggregation the theorem above is an
equality \cite[][Theorem16]{GrandiEndrissAIJ2013}, this is not the
case for database aggregation. This is due to the fact that the
first-order language cannot specify uniquely a given database
instance. The proof of this fact is rather immediate: consider a
dictatorship of the first agent, modified by permuting all the
elements in $U$. That is, $F(\vec{D})=\rho(D_1)$ where $\rho:U\to U$
is any permutation. Clearly, $D_1 \neq \rho(D_1)$, but all constraints
that were satisfied by $D_1$ are also satisfied by $\rho(D_1)$. Hence,
this aggregator is collectively rational with respect to the full
first-order language $\L_{\D}$, but is not a generalised dictatorship.



We now turn our attention to integrity constraints proper to databases. We begin with functional dependencies.

\begin{proposition} \label{prop1}
A quota rule lifts a functional constraint iff $q_P > \frac{n}{2}$ for
all relation symbols $P$ occurring in the functional constraint.
\end{proposition}

\begin{proof}
By assumption, every instance $D_i$ satisfies the constraint.  That is
for every tuple $(u_1,\dots, u_k)$, either there is a unique
$(u_{k+1},\dots, \allowbreak u_q)$ such that
$(u_1,\dots,u_q)=\vec{u}\in D_i(P)$, or there is none.  Suppose now
that the constraint is falsified by the collective outcome.  That is,
there are $\vec{u} \neq \vec{u}'$ such that both $\vec{u}\in
F(\vec{D})(P)$ and $\vec{u}'\in F(\vec{D})(P)$, and
$\vec{u}$ and $\vec{u}'$ coincide on the first $k$ coordinates.  By
definition of quota rules, this means that at least $q_P$ voters are
such that $\vec{u}\in D_i(P)$, and at least $q_P$ possibly different
voters had $\vec{u}'\in D_i(P)$.  Since each individual can have
either $\vec{u}$ or $\vec{u}'$ in $D_i(P)$, by the pigeonhole
principle this is possible if and only if the quota
$q_P \leq \frac{n}{2}$.
\end{proof} 

As immediate applications of Prop.~\ref{prop1}, the intersection rule
clearly lifts any functional dependency, while the union lifts none.
To see the latter, it is sufficient to consider two database instances
that associates different tuples to the same primary key.

\begin{proposition}\label{prop2}
An aggregation procedure $F$ lifts a value constraint if $F$ is
grounded.
\end{proposition}
\begin{proof}
Let $n_k\in D(P_v)$ be a value constraint, where for all $i,j\in \N$,
we have that $D_i(P_v)=D_j(P_v)$.  A grounded aggregation procedure is
such that $F(\vec{D})(P)\subseteq \bigcup_{i \in \N} D_i(P)$. Hence,
for all $\vec{u}\in F(\vec{D})(P)$, there exists an $i\in \N$ such
that $\vec{u}\in D_i(P)$. Since all individual databases satisfy the
value constraint, we have that $u_k\in D_i(P_v)$, and therefore
$u_k\in F(\vec{D})(P_v)\subseteq \bigcup_{i \in \N} D_i(P_v)$, showing
that also $F(\vec{D})(P)$ satisfies the value constraint.
\end{proof}

The converse of the Prop.~\ref{prop2} is not true in general, since a
non-grounded aggregator could be easily devised while still satisfying
a given value constraint.

The last result in this section concerns again quota rules.
\begin{proposition} \label{prop3}
A quota rule lifts a referential constraint $(P_1\to P_2, k)$ iff
$q_{P_2} = 1$.
\end{proposition}

\begin{proof}
Let $\vec{u}\in F(\vec{D})(P_1)$.  Since all the individual databases
satisfy the integrity constraint, we know that for every $i \in \N$
there exists a $\vec{u}_i\in D_i(P_2)$ such that its first $k$
coordinates coincides with the last $k$ coordinates of $P_1$.  Since
all $\vec{u}_i$ are possibly different, they may be supported by one
single individual each. Therefore, the referential constraint is lifted
if and only if the quota relative to $P_2$ is sufficiently small,
i.e., $q_{P_2} = 1$.
\end{proof}

As immediate application of Prop.~\ref{prop3}, intersection and union
rules are included in the results above, since they are quota
rules. As regards distance-based rules, we only remark that they lift
all integrity constraint by their definition, provided that the
minimisation is restricted to consistent databases.

\section{Aggregation and Query Answering}\label{sec:queries}

In this section we analyse one of the most common operation on
databases, i.e., querying, to the light of (rational) aggregation.
Observe that any open formula $\phi(x_1,\dots,x_\ell)$, with free
variables $x_1,\dots,x_\ell$, can be thought of as a
query \cite{AbiteboulHV95}. Evaluating $\phi(x_1,\dots,x_\ell)$ on a
database instance $D$ returns the set $ans(D, \phi)$ of tuples
$\vec{u}=(u_1,\dots,u_\ell)$ such that the assignment $\sigma$, with
$\sigma(x_i) = u_i$ for $i \leq \ell$, satisfies $\phi$, that is,
$(D, \sigma) \models \phi$. Hereafter, with an abuse of notation, we
often write simply $(D, \vec{u}) \models \phi$.
Given the relevance of query answering in database theory, the
following question is of obvious interest.
\begin{question}
What is the relationship between the answer $ans(F(\vec{D}), \phi)$ to
query $\phi$ on the aggregated database $F(\vec{D})$, and the answers
$ans(D_1, \phi), \allowbreak \ldots, ans(D_n, \phi)$ to the same query
on each instance $D_1, \ldots, D_n$?
\end{question}

Clearly, given a query $\phi$, every aggregator $F$ on database
instances induces an aggregation procedure $F^*$ on the query answers,
as illustrated by the following diagram, where $D = F(\vec{D})$:
\begin{center}
\begin{tikzpicture}[auto, node distance=2cm, ->, >=stealth', shorten
>=1pt, semithick]

\tikzstyle{every place/.style}=[fill=white, text=black, minimum size=15pt]
\tikzstyle{every initial by arrow}=[initial text=]
\node 	(s0)   	{$D_1, \ldots, D_n$};
\node    (s1) 	[right of=s0, node distance=5cm]	{$D$};
\node    (s2) 	[below of=s0]	{$ans(D_1, \phi), \ldots, ans(D_n, \phi)$};
\node    (s3) 	[below of=s1]	{$ans(D,\phi)$};
\path 	(s0) 	edge node[above] {$F$} (s1);
\path 	(s0) 	 edge node[left] {$\phi$} (s2);
\path 	(s1) 	 edge node[right] {$\phi$} (s3);
\path 	(s2) 	 edge node[above] {$F^*$} (s3);
\end{tikzpicture}
\end{center}

Hereafter we consider some examples to illustrate this question.
\begin{example} \label{ex1}
If we assume intersection as the aggregation procedure, it is easy to
check that in general the answer to a query in the aggregated database
is not the intersection of the answers for each single instance. To
see this, let $D_1(P) = \{ (a,b) \}$ and $D_2(P) = \{ (a,d) \}$ and
consider query $\phi = \exists y P(x,y)$. Clearly, $ans(D_1 \cap
D_2, \phi)$ is empty, while $ans(D_1, \phi) \cap ans(D_2, \phi)
= \{a\}$.  Hence, in general $\bigcap_{i \in \N}
ans(D_i, \phi) \not \subseteq ans(\bigcap_{i \in \N} D_i, \phi)$.
The converse can also be the case. Consider instances $D_1$,
$D_2$ such that $D_1(P) = \{(a,a), (a,b)\}$, $D_1(R) = \{ c \}$, and
$D_2(P) = \{(a,a), (a,b)\}$, $D_2(R) = \{ d \}$, 
with query $\phi = \forall y P(x,y)$. The intersection
$ans(D_1, \phi) \cap ans(D_2, \phi)$ of answers is empty.
However the answer w.r.t.~the intersection of databases is
$ans(D_1 \cap D_2, \phi) = \{ a \}$, since the active domain of the
intersection only includes elements $a$ and $b$.  As a result, in
general $ans(\bigcap_{i \in \N} D_i, \phi) \not \subseteq
\bigcap_{i \in \N} ans(D_i, \phi)$.
\end{example}

Similar arguments can be used to show that the union of answers is in
general different from the answer on the union of instances.

These examples shows that it is extremely difficult to find
aggregators that commute for any first-order query
$\phi \in \L_{\D}$. Hence, they naturally raise the question of
syntactic restrictions on queries such that the aggregation procedure
$F^* = \phi \circ F \circ \phi^{-1}$ on answers can be expressed
explicitly in terms of $F$ (e.g., the intersection of answers is the
answer to the query on the intersection):
\begin{question} \label{quest1}
Given aggregation procedures $F$ and $F^*$, is there a restriction of
the query language for $\phi$ such that the diagram above commute?
\end{question}

This problem is related to the following, more general question.
\begin{question} \label{quest2}
  Given an aggregation procedure $F$ and a query 
  language $\L$, what is the aggregation procedure $F^*$?  Can $F^*$
be represented explicitly?
\end{question}


The following result provides a first, partial answer to
Question~\ref{quest1}, in the case $F$ and $F^*$ are unions.
\begin{lemma}[Existential Fragment] \label{existential}
Consider the positive existential fragment $\L^+_{\exists}$ of
first-order logic defined as follows:
\begin{eqnarray*}
\phi  & ::= &  P(x_1, \ldots, x_q) \mid 
\phi \lor \phi \mid \exists x \phi
\end{eqnarray*}
The language $\L^+_{\exists}$ is lifted by unions, that is, for $F$
and $F^*$ equal to set-theoretical union, the diagram commutes for the
query language $\L^+_{\exists}$.
\end{lemma}

\begin{proof}
The proof is by induction on the structure of query $\phi$.
For atomic $\phi = P(x_1, \ldots, x_q)$, $\vec{u} \in
ans(\bigcup_{i \in \N} D_i,\phi)$ iff $(\bigcup_{i \in \N}
D_i, \vec{u}) \models \phi$, iff for some $i \in \N$, $(D_i,
\vec{u}) \models \phi$, iff $\vec{u} \in ans(D_i,\phi)$  for some $i \in \N$,
iff $\vec{u} \in \bigcup_{i \in \N} ans(D_i,\phi)$.


For $\phi = \psi \lor \psi'$, $\vec{u} \in ans(\bigcup_{i \in \N}
  D_i,\phi)$ iff $(\bigcup_{i \in \N} D_i, \vec{u}) \models \phi$, iff
  $(\bigcup_{i \in \N} D_i, \vec{u}) \models \psi$ or
  $(\bigcup_{i \in \N} D_i, \vec{u}) \models \psi'$, iff for some $i,
  j \in \N$, $(D_i, \vec{u}) \models \psi$ or
  $(D_j, \vec{u}) \models \psi'$ by induction hypothesis. In
  particular, we have both $(D_i, \vec{u}) \models \psi \lor \psi'$
  and $(D_j, \vec{u}) \models \psi \lor \psi'$, that is,
  $\vec{u} \in \bigcup_{i \in \N} ans(D_i,\phi)$.  On the other hand,
  $\vec{u}
\in \bigcup_{i \in \N} ans(D_i,\phi)$ iff $\vec{u} \in ans(D_i,\phi)$
for some $i \in \N$, iff $(D_i, \vec{u}) \models \psi$ or $(D_i,
\vec{u}) \models \psi'$. In both cases, by induction hypothesis $(\bigcup_{i \in \N} D_i,
\vec{u}) \models \phi$, that is, $\vec{u} \in ans(\bigcup_{i \in \N}
  D_i,\phi)$.

For $\phi = \exists x \psi$, $\vec{u} \in ans(\bigcup_{i \in \N}
D_i,\phi)$ iff $(\bigcup_{i \in \N} D_i, \vec{u}) \models \phi$, iff
for some $u \in adom(\bigcup_{i \in \N} D_i)$, $(\bigcup_{i \in \N}
D_i, \vec{u} \cdot u) \models \psi$, and therefore for some $i,
j \in \N$, $u \in adom(D_j)$ and $(D_i, \vec{u} \cdot
u) \models \psi$.  Notice that if $(D_i, \vec{u} \cdot
u) \models \psi$, then $u \in adom(D_i)$ as well, as $\phi$ belongs to
the positive (existential) fragment of first-order logic.
Hence, for some $i \in \N$, $u \in adom(D_i)$ and $(D_i, \vec{u} \cdot
u) \models \psi$,  that is,
$\vec{u} \in ans(D_i,\phi)$ for some $i \in \N$. 
On the other hand, $\vec{u} \in \bigcup_{i \in \N} ans(D_i,\phi)$ iff
$\vec{u} \in ans(D_i,\phi)$ for some $i \in \N$, iff $u \in adom(D_i)$
and $(D_i, \vec{u} \cdot u) \models \psi$, that is, $u \in
adom(\bigcup_{i \in \N} D_i)$ and $(\bigcup_{i \in \N}
D_i, \vec{u} \cdot u) \models \psi$ by induction hypothesis. Hence,
$\vec{u} \in ans(\bigcup_{i \in \N} D_i,\phi)$.
\end{proof}

By Lemma~\ref{existential} queries in $\L^+_{\exists}$ are preserved
whenever both $F$ and $F^*$ are unions. Obviously, it would be of
interest to find what is the largest fragment $\L'$ of first-order
logic preserved by unions. By the results in this section we know that
$\L^+_{\exists} \subseteq \L' \subset \L_{\D}$.

Further, we may wonder whether a result symmetric to
 Lemma~\ref{existential} holds for intersections and the positive
 universal fragment $\L^+_{\forall}$ of first-order logic defined as
 follows:
\begin{eqnarray*}
\phi  & ::= &  P(x_1, \ldots, x_q) \mid 
 \phi \land \phi \mid \forall x \phi
\end{eqnarray*}

Unfortunately, in Example~\ref{ex1} we provided a formula $\phi
= \forall y P(x, y)$ in $\L^+_{\forall}$ and instances $D_1$, $D_2$
such that $ans(D_1 \cap D_2, \phi) \not \subseteq ans(D_1, \phi) \cap
ans(D_2, \phi)$. Hence, for $F$ and $F^*$ equal to set-theoretical
intersection, the diagram above does not commute for the query
language $\L^+_{\forall}$.

Nonetheless, we are able to prove a weaker but still significant
result related to Question~\ref{quest2}. Specifically, the next lemma
shows that
if in the diagram above $F$ is the intersection and the query language
is $\L^+_{\forall}$, then $F^*$ is unanimous, in the sense that
$\bigcap_{i \in \N} ans(D_i, \phi) \subseteq ans(\bigcap_{i \in \N}
D_i),\phi)$.
\begin{lemma} \label{universal}
Let the aggregator $F$ be the intersection and let the query language
be $\L^+_{\forall}$. Then, the lifted aggregator $F^*$ is unanimous.
\end{lemma}
\begin{proof}
  We prove that $\bigcap_{i \in \N} ans(D_i, \phi) \subseteq
  ans(\bigcap_{i \in \N} D_i,\phi)$.  So, if
  $\vec{u} \in \bigcap_{i \in \N} ans(D_i, \phi)$ then for every
  $i \in \N$, $(D_i, \vec{u}) \models \phi$.  We now prove by
  induction on $\phi \in \L^+_{\forall}$ that if for every $i \in \N$,
  $(D_i, \vec{u}) \models \phi$, then $(\bigcap_{i \in \N}
  D_i, \vec{u}) \models \phi$.  As to the base case for $\phi =
  P(x_1, \ldots, x_q)$ atomic, $(D_i, \vec{u}) \models P(x_1, \ldots,
  x_q)$ iff $\vec{u} \in D_i(P)$ for every $i \in \N$. In particular,
  $\vec{u} \in \bigcap_{i \in \N} D_i(P)$ as well, and therefore
  $(\bigcap_{i \in \N} D_i, \vec{u}) \models P(x_1, \ldots, x_q)$.  As
  to the inductive case for $\phi = \psi \land \psi'$, suppose that
  $(D_i, \vec{u}) \models \phi$, that is,
  $(D_i, \vec{u}) \models \psi$ and $(D_i, \vec{u}) \models \psi$ for
  every $i \in \N$. By induction hypothesis we obtain that
  $(\bigcap_{i \in \N} D_i, \vec{u}) \models \psi$ and
  $(\bigcap_{i \in \N} D_i, \vec{u}) \models \psi'$, i.e.,
  $(\bigcap_{i \in \N} D_i, \vec{u}) \models \phi$.  Finally, if
  $(D_i, \vec{u}) \models \forall x \psi$ for every $i \in \N$, then
  for all $v \in \adom(D_i)$, $(D_i, \vec{u} \cdot
  v) \models \psi$. In particular, for all
  $v \in \adom(\bigcap_{i \in \N} D_i)$, $(D_i, \vec{u} \cdot
  v) \models \psi$ for every $i \in \N$, and by induction hypothesis,
  for all $v \in \adom(\bigcap_{i \in \N} D_i)$, $(\bigcap_{i \in \N}
  D_i, \vec{u} \cdot v) \models \psi$, i.e., $(\bigcap_{i \in \N}
  D_i, \vec{u}) \models \forall x \psi$.  As a result, $\vec{u} \in
  ans(\bigcap_{i \in \N} D_i, \phi)$.
\end{proof}

A result symmetric to Lemma~\ref{universal} holds for language
 $\L^+_{\exists}$ and unions:
\begin{lemma} \label{existential2}
Let the aggregator $F$ be the union and let the query language
be $\L^+_{\exists}$. Then, the lifted aggregator $F^*$ is grounded.
\end{lemma}
\begin{proof}
  We prove that $ans(\bigcup_{i \in \N} D_i,\phi)  \subseteq
  \bigcup_{i \in \N} ans(D_i, \phi)$.  So, if
  $\vec{u} \in ans(\bigcup_{i \in \N} D_i,\phi)$ then 
$(\bigcup_{i \in \N} D_i, \vec{u}) \models \phi$.  We now prove by
  induction on $\phi \in \L^+_{\forall}$ that if $(\bigcup_{i \in \N}
  D_i, \vec{u}) \models \phi$, then for some $i \in \N$,
  $(D_i, \vec{u}) \models \phi$.  As to the base case for $\phi =
  P(x_1, \ldots, x_q)$ atomic, $(\bigcup_{i \in \N}
  D_i, \vec{u}) \models P(x_1, \ldots, x_q)$ iff
  $\vec{u} \in \bigcup_{i \in \N} D_i(P)$, iff $\vec{u} \in D_i(P)$
  for some $i \in \N$. In particular, $(D_i, \vec{u}) \models
  P(x_1, \ldots, x_q)$ as well, and therefore
  $\vec{u} \in \bigcup_{i \in \N} ans(D_i, \phi)$.  As to the
  inductive case for $\phi = \psi \lor \psi'$, suppose that
  $(\bigcup_{i \in \N} D_i, \vec{u}) \models \phi$, that is,
  $(\bigcup_{i \in \N} D_i, \vec{u}) \models \psi$ or
  $(\bigcup_{i \in \N} D_i, \vec{u}) \models \psi'$. In the first
  case, by induction hypothesis we have that
  $\vec{u} \in \bigcup_{i \in \N} ans(D_i, \psi)$
  i.e., for some $i \in \N$, $(D_i, \vec{u}) \models \psi$, and
  therefore $(D_i, \vec{u}) \models \phi$. Hence, $\vec{u} \in
  ans(D_i, \phi)$ for some $i \in \N$, that is,
  $\vec{u} \in \bigcup_{i \in \N} ans(D_i, \phi)$. The case for
  $(\bigcup_{i \in \N} D_i, \vec{u}) \models \psi'$ is symmetric.
  Finally, if
  $(\bigcup_{i \in \N} D_i, \vec{u}) \models \exists x \psi$, then
  for some $v \in \adom(\bigcup_{i \in \N} D_i)$, $(\bigcup_{i \in \N} D_i, \vec{u} \cdot
  v) \models \psi$. In particular, 
by induction hypothesis, $\vec{u} \cdot v \in \bigcup_{i \in \N}
ans(D_i, \psi)$, that is, $( D_i, \vec{u} \cdot v) \models \psi$ for
some $i \in \N$. Further, since $\psi$ is a positive formula,
$v \in \adom(D_i)$, and therefore, $(D_i, \vec{u} \cdot
v) \models \phi$, i.e., $\vec{u} \in \bigcup_{i \in \N}
ans(D_i, \phi)$.
\end{proof}

To conclude this section we discuss the results obtain so far.  We
said that Lemma~\ref{existential} can be seen as a (partial) answer to
Question~\ref{quest1}. Similarly, Lemma~\ref{universal}
and \ref{existential2} are related to Question~\ref{quest2}.  Results
along the lines of Lemmas~\ref{existential}-\ref{existential2} may
find application in efficient query answering: it might be that in
some cases, rather than querying the aggregated database $F(\vec{D})$,
it is more efficient to query the individual instances $D_1, \ldots,
D_n$ and then aggregate the answers.  In such cases it is crucial to
know which answers are preserved by the different aggregation
procedures. The results provided in this section aimed to be a first,
preliminary step in this direction.

%





\section{Conclusions and Related Work}\label{sec:conclusions}

In this paper we have proposed a framework for the aggregation of
conflicting information coming from multiple sources in the form of
finite relational databases.  We proposed a number of aggregators
inspired by the literature on social choice theory, and adapted a
number of axiomatic properties.  We then focused on two natural
questions which arise when dealing with the aggregation of
databases. First, we studied what languages for integrity constraints
are lifted by some of the rules we proposed, i.e., what constraints
are true in the aggregated database supposing that all individual
input satisfies the same constraints.  Second, we investigated
first-order query answering in the aggregated databases,
characterising some languages for which the aggregation of the answers
to the individual databases corresponds to the answer to the query on
the aggregated database.

Our initial results shed light on the possible use of choice-theoretic
techniques in the database merging and integration, and opens multiple
interesting directions for future research.  In particular, the
connections to the literature on aggregation and merging can be
investigated further.  First, section~\ref{sec:characterisations}
showcased results for which database aggregation behaves similarly to
binary aggregation with integrity constraints
(see \cite{GrandiEndrissAIJ2013}), but pointed out at some crucial
differences.  In particular, there are natural classes of integrity
constraints used in databases for which the equivalent in
propositional logic, the language of choice for binary aggregation,
would be tedious and lenghty. We were able to provide initial results
on their preservation through aggregation.  Second, the recent work of
Endriss and Grandi \cite{EndrissGrandiAIJ2017} is also strongly
related to our contribution. Since graphs are a specific type of
relational structures, our work directly generalise their graph
aggregation framework to relations of arbitrary arity. However, the
specificity of their setting allows them to obtain very powerful
impossibility results, which are yet to be explored in the area of
database aggregation.  Third, to the best of our knowledge the problem
of aggregated query answering is new in the literature on aggregation,
albeit a similar problem has been studied in the aggregation of
argumentation graphs \cite{ChenEndrissTARK2017}, a setting closer to
that of graph aggregation. Also this direction deserves further
investigation.

\end{document}